\documentclass[article,lineno]{article}

\usepackage{amsmath}
\usepackage{times}
\usepackage{bm}
\usepackage[round]{natbib}
\usepackage[plain,noend]{algorithm2e}
\usepackage{comment}
\usepackage{amssymb}
\usepackage{graphicx}
\usepackage{caption}
\usepackage{bbold}

\usepackage{amsthm}

\makeatletter
\renewcommand{\algocf@captiontext}[2]{#1\algocf@typo. \AlCapFnt{}#2} 
\def\@algocf@capt@plain{top}
\renewcommand{\algocf@makecaption}[2]{%
  \addtolength{\hsize}{\algomargin}%
  \sbox\@tempboxa{\algocf@captiontext{#1}{#2}}%
  \ifdim\wd\@tempboxa >\hsize
    \hskip .5\algomargin%
    \parbox[t]{\hsize}{\algocf@captiontext{#1}{#2}}
  \else%
    \global\@minipagefalse%
    \hbox to\hsize{\box\@tempboxa}
  \fi%
  \addtolength{\hsize}{-\algomargin}%
}
\makeatother

\def\T{{ \mathrm{\scriptscriptstyle T} }}

\newtheorem{theorem}{Theorem}
\newtheorem{lemma}{Lemma}

\begin{document}

\title{General Bayesian Updating and the Loss-Likelihood Bootstrap}

\author{S. P. Lyddon, C. C. Holmes, S. G. Walker}
\date{}

\maketitle

\begin{abstract}
In this paper we revisit the weighted likelihood bootstrap, a method that generates samples from an approximate Bayesian posterior of a parametric model. We show that the same method can be derived, without approximation, under a Bayesian nonparametric model with the parameter of interest defined as minimising an expected negative log-likelihood under an unknown sampling distribution. This interpretation enables us to extend the weighted likelihood bootstrap to posterior sampling for parameters minimizing an expected loss. We call this method the loss-likelihood bootstrap. We make a connection between this and general Bayesian updating, which is a way of updating prior belief distributions without needing to construct a global probability model, yet requires the calibration of two forms of loss function. The loss-likelihood bootstrap is used to calibrate the general Bayesian posterior by matching asymptotic Fisher information. We demonstrate the methodology on a number of examples.
\end{abstract}

\section{Bayesian inference and model misspecification}\label{sec:intro}

Bayesian theory provides a comprehensive framework for quantitative reasoning about uncertainty which is built on axiomatic foundations and is well suited to many modern scientific applications. For the statistician, much of the validity of the Bayesian approach hinges on the condition that the statistical model for the data is well specified, in that it contains the underlying data-generating mechanism. This cannot be expected to hold in general, if at all, and the extent and impact of the misspecification can be hard to quantify.

In this article we investigate statistical methods that remain honest about the inability to perfectly model the data. Suppose we observe a sample $x_{1:n} = ( x_1,\ldots, x_n )$ with the $x_i \in \Omega \subseteq \mathbb{R}^p$ being independent draws from an unknown distribution $F_0$, which we assume admits a density, $f_0$. In a parametric Bayesian modelling paradigm, a family of densities $\mathbb{F}_\Theta~=~\{ p(~\cdot~\mid~\theta ) ; \, \theta \in \Theta \subseteq \mathbb{R}^d \}$ is specified along with a prior belief distribution $p(\theta)$ for the unknown parameter $\theta$. Assuming that $f_0 \in \mathbb{F}_\Theta$, beliefs about the true parameter are updated by conditioning on the observed data, using Bayes' rule; $p(\theta \mid x_{1:n} ) \, \propto \, p(\theta) \prod_{i=1}^n p( x_i \mid \theta)$.

We say that such a model is well specified if there exists a $\theta_0 \in \Theta$ supported by the prior such that $f_0(\cdot) = p( \cdot \mid \theta_0)$. If this is the case, then under mild regularity conditions the posterior will concentrate at $\theta_0$ as the number of observations increases. Conversely, if $f_0 \notin \mathcal{F}$, we say that the model is misspecified. In this case the posterior will, under mild regularity conditions, concentrate at the pseudo-true parameter value which minimizes the Kullback--Leibler divergence \citep{Kullback1951} to the true sampling distribution \citep{Berk1966a}; i.e.

\begin{equation*}\label{eq:bayes_argmin}
\theta_0 = \arg \min_{\theta \in \Theta} \, \int \log \left\{ \frac{f_0(x)}{p( x \mid \theta) } \right\} f_0(x) dx 
= \arg \max_{\theta \in \Theta} \, \int \log p( x \mid \theta) \, dF_0(x).
\end{equation*}
\citet{Muller2013} showed that, under regularity conditions, the asymptotic frequentist risk associated with misspecified Bayesian estimators is inferior to that of an artificial posterior which is normally distributed, centred at the maximum likelihood estimator and with a certain covariance matrix. Additionally, the rate at which the Bayesian learns about the parameter of a misspecified model is not necessarily optimal, as it is in the well-specified case; see \cite{Zellner1988}. 

It is therefore important that the statistician accounts for the inadequacy of the model when making inference about $\theta_0$ and we shall study ways to do this. We show that the weighted likelihood bootstrap \citep{Newton1994}, a method originally designed for approximate sampling from the Bayesian posterior of a well specified parametric model, can also be viewed as generating exact posterior samples under a Bayesian nonparametric model which assumes much less about the structure of the data-generating mechanism than the parametric model does. This approach can then be extended to generate posterior samples for a much larger family of parameters than just those indexing parametric models.

More often than not, when a statistical model is used, it is with a specific action in mind, such as a prediction. For the Bayesian, the optimal action is chosen by maximising an expected utility or, equivalently, minimising an expected loss. To do this it is necessary to construct a full probability model for the data. This can be a prohibitively expensive exercise, if it is indeed at all possible. It is natural, then, to question whether this global modelling approach is necessary, or whether we could instead focus solely and directly on the functional that is of interest.

\citet{Bissiri2016} presented a general framework for updating targeted belief distributions of this kind. Instead of restricting themselves to parameters that index a family of distribution functions, the authors considered general parameters $\theta$ whose true value $\theta_0$ minimizes an expected loss for some loss function $\ell: \Theta \times \Omega \rightarrow [0,\infty)$; 
\begin{equation}\label{eq:argmin}
\theta_0 = \theta(F_0) = \arg \min_{\theta \in \Theta} \, \int \ell( \theta, x ) \,dF_0(x).
\end{equation}
Treating $F_0$ as unknown, \citet{Bissiri2016} use a decision-theoretic argument, relying on a coherency property, that leads to a unique functional form for the update of prior beliefs $p(\theta)$, given observations $x_{1:n}$. The resulting posterior distribution $p_{GB,w}$ is fully determined up to a loss scale $w > 0$, and given by
\begin{equation}\label{eq:gb_update}
p_{GB,w}(\theta \mid x_{1:n}) \, \propto \, p(\theta) \exp \{ - w \ell( \theta, x_{1:n} ) \},
\end{equation}
where the loss for multiple observations is defined additively, i.e. $\ell(\theta, x_{1:n}) = \sum_{i=1}^n \ell(\theta, x_i)$. We assume throughout this paper that this posterior distribution is proper; i.e. the loss and prior are provided such that the right-hand side of (\ref{eq:gb_update}) is integrable.

We refer to this posterior distribution as the general Bayesian posterior, because it can be computed for a more general family of parameters than those simply indexing a family of distributions. Bayes rule is recovered as a special case by choosing $w=1$ and using the self-information loss, $\ell(\theta,x)= - \log p(x \mid \theta)$. We call the data-dependent component $\exp \{ -w \ell(\theta, x ) \}$ the loss likelihood as it provides a prior-to-posterior belief update for the parameter $\theta$, in analogy with the likelihood in the Bayesian setting. The loss scale $w$ is a non-negative scalar controlling the learning rate about $\theta$ attributable to the observed data. If the learning rate is too large, the posterior will be too concentrated, exaggerating the extent of the information about $\theta$ coming from the data. Conversely if $w$ is too small, the posterior will underplay the  information in the data relative to the prior.

The theory presented in \citet{Bissiri2016} has a number of attractions compared to the Bayesian approach. It is built upon the assumption that the true underlying data-generating mechanism is unknown. It provides a principled means for performing targeted prior belief updates without the burden of having to construct a global probabilistic model. The prior specification is local to the parameter of interest. However, although a number of suggestions have appeared in the literature, the setting $w$ remains an open problem. We propose that a Bayesian bootstrap should be used for this purpose and we extend the weighted likelihood bootstrap's interpretation to cover the wider class of models built from loss functions, of the form given in (\ref{eq:argmin}). We refer to this method as the loss-likelihood bootstrap. The asymptotic structure of this bootstrap is studied, alongside that of the general Bayesian  posterior. These theoretical results are used to determine a loss scale for the general Bayesian posterior by matching asymptotic posterior information. This provides us with a calibrated general Bayesian posterior with a number of desirable properties, such as the Bayes posterior being recovered if the model is well specified. The general Bayesian approach admits a subjective prior and provides a prior-to-posterior belief update, which in some settings will be preferable over the prior-free loss-likelihood bootstrap. Although, as noted in \cite{Newton1994}, the weighted likelihood bootstrap can provide a better approximation to the Bayesian posterior than a normal approximation, asymptotically, if the prior used is the square of the Jeffreys prior.


\section{Revisiting the weighted likelihood bootstrap}\label{sec:wlb}
The weighted likelihood bootstrap \citep{Newton1994} is a method for approximately sampling from a posterior distribution of a well specified parametric statistical model. Samples are generated by computing randomly-weighted maximum likelihood estimates. The weights are drawn from a Dirichlet distribution, scaled by the number of observations. The weights perturb the contribution of each observation to the likelihood; see Algorithm \ref{algo:wlb} for details.

\begin{algorithm}[!h]
\vspace*{-6pt}
\caption{The Weighted Likelihood Bootstrap} \label{algo:wlb}
\vspace*{-9pt}
\begin{tabbing}
   \enspace For $j=1$ to $j=B$: \\
   \qquad Draw random weights $(g_{j1},\ldots,g_{jn})$ with \\ \qquad \qquad $n^{-1}(g_{j1},\ldots,g_{jn}) \sim \mathrm{Dirichlet}(1,\ldots,1)$\\
   \qquad Compute $\theta^{(j)} = \arg \max_{\theta \in \Theta} \prod_{i=1}^n p( x_i \mid \theta )^{g_{ji}}$ \\
\enspace Output $\left( \theta^{(1)},\ldots,\theta^{(B)} \right)$
\end{tabbing}
\vspace*{-6pt}
\end{algorithm}

The method produces independent samples and is trivially parallelizable over $j=1,\ldots,B$, which is advantageous over Markov chain Monte Carlo methods. However, the weighted likelihood bootstrap is not an exact method for sampling from the  posterior of a parametric model and does not accommodate a prior. However, it is asymptotically first-order equivalent to a Bayesian posterior if the parametric model is well specified, with a higher order of asymptotic equivalence, under certain conditions, if the prior used is the square of the Jeffreys prior. 

If interest is in the $\theta$ minimizing $-\int \log p( x \mid \theta) \, dF_0(x)$ then it is possible to arrive at a method functionally equivalent to the weighted likelihood bootstrap, but without assuming the parametric model is well-specified. Under misspecification, 
the data-generating distribution function is unknown and so it is appropriate to construct a prior on the sampling distribution function $F$, whose unknown true value is $F_0$. Uncertainty about $\theta$ is inherited from uncertainty about $F$. More precisely, a prior on $F$, say $P_{\mathcal{F}}$, induces a prior probability $P_\Theta$ on $\Theta$, via 
\begin{equation}\label{eq:def_theta}
\theta(F)=\arg\max_{\theta\in \Theta}\int \log p(x\mid\theta)\,d F(x);
\end{equation} 
the induced prior being
$P_\Theta(\theta \in A)=P_{\mathcal{F}}(\{F:\,\theta(F)\in A\})$.

The Bayesian nonparametric literature provides a number of methods for constructing priors over the space of distribution functions. The Dirichlet process prior \citep{Ferguson1973} is a natural choice, as it is simple to use. The hyperparameter that determines the Dirichlet process prior is a finite measure $\alpha$. Upon observing $x_{1:n}$ the posterior is also a Dirichlet process with unit mass added to the base measure at each observation; i.e.
\begin{equation*}
{\cal L}\left(F \mid x_{1:n}\right) = \mathrm{DP}\left( \alpha + \sum\limits_{i=1}^n \delta_{x_i} \right).
\end{equation*}
For a detailed account of the Dirichlet process, see \cite{ghosal2017fundamentals}. Under regularity conditions the posterior for $F$ will concentrate at $F_0$. In a noninformative setting, a small value of $\alpha(\Omega)$ is chosen. In the limit $\alpha(\Omega) \to 0$ the posterior distribution is supported only by the observations $x_{1:n}$, with Dirichlet-distributed probabilities for each state. This sampling procedure is commonly referred to as the Bayesian bootstrap \citep{Rubin1981} as it has a limiting Bayesian interpretation and is closely associated to Efron's bootstrap \citep{Efron1979}. Posterior sampling under the Bayesian bootstrap is direct and fast and trivially parallelizable; see Algorithm \ref{algo:bayes_bootstrap} for further details. 

\begin{algorithm}[!h]
\vspace*{-6pt}
\caption{The Bayesian Bootstrap} \label{algo:bayes_bootstrap}
\vspace*{-9pt}
\begin{tabbing}
   \enspace For $j=1$ to $j=B$: \\
   \qquad Draw random distribution $F^{(j)} =\sum_{i=1}^n g_{ji}\,\delta_{x_i}$ with \\ \qquad \qquad  $(g_{j1},\ldots,g_{jn}) \sim \mathrm{Dirichlet}(1,\ldots,1)$\\
   \qquad Compute $\theta^{(j)} = \theta \left(F^{(j)}\right)$ as in (\ref{eq:def_theta}) \\
\enspace Output $\left( \theta^{(1)},\ldots,\theta^{(B)} \right)$
\end{tabbing}
\vspace*{-6pt}
\end{algorithm}

If we apply the Bayesian bootstrap strategy to $\theta(F)$ as defined in (\ref{eq:def_theta}), we recover the weighted likelihood bootstrap. Conceptually, however, the parametric and nonparametric constructions are quite different. This can perhaps most easily be seen through the role that the random weights play. For the weighted likelihood bootstrap, the weights perturb the contribution of each observation to the likelihood, whereas for the Bayesian bootstrap the weights represent a posterior sample of the unknown distribution function.

The Bayesian bootstrap can produce posterior samples for a much larger family of parameters than just those that maximize expected log likelihoods, see for example \cite{Chamberlain2003}. Of interest to us is the family of parameters than minimize an expected loss, as in (\ref{eq:argmin}). We call the method of posterior sampling for this family of parameters the loss-likelihood bootstrap; see Algorithm \ref{algo:llb} for details. The reason for giving this name is that the method will prove important when we consider the problem of calibrating general Bayesian posterior distributions in Section \ref{sec:gb_calib}.

\begin{algorithm}[!h]
\vspace*{-6pt}
\caption{The Loss-likelihood Bootstrap} \label{algo:llb}
\vspace*{-9pt}
\begin{tabbing}
   \enspace For $j=1$ to $j=B$: \\
   \qquad Draw random distribution $F^{(j)} =\sum_{i=1}^n g_{ji}\,\delta_{x_i}$ with \\ \qquad \qquad $(g_{j1},\ldots,g_{jn}) \sim \mathrm{Dirichlet}(1,\ldots,1)$\\
   \qquad Compute $\theta^{(j)} = \arg \min_{\theta \in \Theta} \int \ell(\theta,x) dF^{(j)}(x)$ \\
\enspace Output $\left( \theta^{(1)},\ldots,\theta^{(B)} \right)$
\end{tabbing}
\vspace*{-6pt}
\end{algorithm}

Asymptotic properties of the weighted likelihood bootstrap have been studied in the 1991 University of Washington PhD thesis by M.A. Newton. It was shown that if the parametric model is well-specified then the distribution of samples under the weighted likelihood bootstrap is asymptotically first-order correct to a Bayesian posterior distribution. By this we mean that the probability laws ${\cal L}\{n^{1/2} ( \theta - \hat{\theta}_n ) \mid x_{1:n}\}$ and ${\cal L}\{n^{1/2} ( \tilde{\theta}_n - \hat{\theta}_n ) \mid x_{1:n}\}$ converge  to the same limit as $n \rightarrow \infty$, almost surely with respect to $x_{1:n} \sim p(\cdot \mid \theta_0)$, where $\hat{\theta}_n$ is the maximum likelihood estimator, $\tilde{\theta}_n$ is a random sample from the weighted likelihood bootstrap, and $\theta$ represents the parameter under a Bayesian posterior.

The proof proceeds by determining the relevant asymptotic properties for both the weighted likelihood bootstrap under misspecification, and also for the loss-likelihood bootstrap. 
The relevant generalized asymptotic theory is contained in the following theorem:

\begin{theorem}\label{theorem1}
Let $\tilde{\theta}_n$ be a loss-likelihood bootstrap sample of a parameter defined in (\ref{eq:argmin}) with loss function $\ell$, given $n$ observations $x_{1:n}$, and let $P_{LL}$ be its probability measure. Under regularity conditions, for any Borel set $A \subset \mathbb{R}^d$, as $n \rightarrow \infty$ we have
\begin{equation*}
P_{LL} \left\{ n^{1/2} \left( \tilde{\theta}_n - \hat{\theta}_n \right) \in A  \mid x_{1:n} \right\} \ \rightarrow \ P( z \in A),
\end{equation*}  
a.s. $x_{1:\infty}$, where  $z \sim N_d\{ \, 0, \, J( \theta_0)^{-1} I(\theta_0) J(\theta_0)^{-1} \, \}$, with 
$$
I(\theta) = \int_{\Omega} \nabla \ell ( \theta, x) \nabla \ell ( \theta, x)^\T \,dF_0(x) \quad\mbox{and}\quad
J(\theta) = \int_{\Omega} \nabla^2 \ell( \theta, x) \,dF_0(x),
$$
where $\nabla$ is the gradient operator with respect to $\theta$, and $\hat{\theta}_n = \arg \min_\theta  n^{-1} \sum_{i=1}^n \ell( \theta , x_i )$.
\end{theorem}

\begin{proof}
The proof follows along the lines of the weighted likelihood bootstrap asymptotic normality proof in the 1991 University of Washington PhD thesis by M.A. Newton. Details can be found in the Supplementary Material.
\end{proof}

The asymptotic covariance matrix in Theorem \ref{theorem1} is a well-known quantity in the robust statistics literature; sometimes called the sandwich covariance matrix. It was shown in \citet{Huber1967} to be the asymptotic covariance matrix for general, potentially misspecified, maximum likelihood estimators. This asymptotic distribution and that of the Bayesian posterior do not coincide if the model is misspecified, in general. \cite{Muller2013} showed that the sandwich covariance matrix can lead to an improvement in frequentist risk over a misspecified Bayesian posterior. Others, such as \citet{Royall2003}, have argued that the  sandwich covariance matrix can be used to make misspecified likelihood functions robust.

General Bayesian models admit a prior distribution and provide a belief update, for the same type of parameters as the loss-likelihood bootstrap, which does not admit a prior. However, the general Bayesian loss scale, $w$ in (\ref{eq:gb_update}), must be calibrated. Under regularity conditions, the impact of a prior diminishes as the number of observations grows large. This indicates a potential strategy for calibrating general Bayesian posterior distributions to the loss-likelihood bootstrap, by ensuring that asymptotically these posteriors contain the same amount of information. We develop this idea in the next section.

\section{Calibrating general Bayesian posteriors by asymptotic covariance matching}\label{sec:gb_calib}

The loss-likelihood bootstrap posterior and the general Bayesian posterior have much in common; they target the same parameter; i.e. the $\theta_0$ in (\ref{eq:argmin}), and are not parametric, though to differing degrees. Thus, for large samples, as the data dominates the posterior, we would expect these distributions to be comparable with respect to their asymptotic normal distributions. Hence, we seek to match these two methods via these distributions, which will then provide a means by which to specify $w$ to match the information in the data.

The general Bayesian posterior of (\ref{eq:gb_update}) has an asymptotic normal distribution, under regularity conditions. A second-order Taylor expansion of the loss function about the empirical risk minimizer provides some intuition about the nature of this distribution. If the minimizing parameter value is in the interior of the parameter space then the first derivative of the sample loss evaluated at the empirical risk minimizer is zero; i.e.
$\nabla \ell (\hat{\theta}_n, x_{1:n} )  = 0.$

The second-order Taylor approximation about $\hat{\theta}_n$ is as follows,
\begin{equation*}
\ell( \theta, x_{1:n} ) = \ell( \hat{\theta}_n, x_{1:n} ) \, + \, \frac{1}{2} (\theta - \hat{\theta}_n)^\T  \, \nabla^2 \ell( \hat{\theta}_n, x_{1:n} )  (\theta - \hat{\theta}_n) \,+ \, o_P( \, n  \,\vert \theta - \hat{\theta}_n \vert^2 \,).
\end{equation*}
Using this approximation in place of $\ell$ in (\ref{eq:gb_update}), we would expect for regular models that
\begin{equation}\label{gbw}
n^{1/2}(\theta-\hat{\theta}_n)\to z'\quad \mbox{in distribution,} \quad \mbox{a.s.} \ \ F^{\infty}_0, 
\end{equation}
where  $z' \sim N_d( 0, w^{-1} J( \theta_0)^{-1}  )$. Regularity conditions and a proof of this result can be found in \citet{Chernozhukov2003}; details can also be found in the Supplementary Material.

Thus under regularity conditions, both the loss likelihood and general Bayesian posterior distributions are asymptotically normal, with the same centering and scaling but different covariance matrices. The Fisher information matrix is a well-understood measure of the information in a sample relating to a parameter. Asymptotically, the posteriors can be considered as normal location models centred at the maximum likelihood estimator. However, the Fisher information matrix is clearly a matrix, and we have a scalar $w$ which can be used to calibrate the general Bayesian distribution.

\cite{Ferentinos1981} considered the problem of constructing one-dimensional information metrics from the Fisher information matrix, and argued that such metrics should be non-negative and strictly increasing functions of its eigenvalues, to ensure the resulting metric satisfies a number of properties. Two natural choices are the trace and determinant of the Fisher information matrix, which equate to the sum and product of the eigenvalues, respectively. This coincides with some quantities known in information theory; the differential entropy of a normal distribution is a function of the determinant of the precision matrix, and the less well-known Fisher information number, sometimes referred to as simply the Fisher information for a density, is equal to the trace of the Fisher information matrix. 


In this work we choose the Fisher information number, i.e. the trace of the Fisher information matrix, as our information metric primarily due to its simplicity in computation. It takes the value zero for a flat posterior, and is positive otherwise, which is not true of differential entropy. It is the sum of the marginal Fisher information for each dimension, which is a well-understood quantity in statistics that summarizes the amount of information in a sample about a parameter. We shall denote the Fisher information number as $K(p)$, which is defined as follows,
\begin{equation*}
K(p) = \int \frac{  \vert \nabla p(\theta) \vert^2}{p(\theta)}d\theta.
\end{equation*}
\cite{Walker2016} showed that the Fisher information number can be used to measure the information in a Bayesian experiment, in analogy with the work of \cite{Lindley1956} who used differential entropy. Further, \cite{Holmes2017} used the Fisher information number to calibrate a power likelihood temperature; see the discussion section for more details on this.

The following lemma determines the loss scale required to match the Fisher information number of the general Bayesian posterior to the loss-likelihood bootstrap.

\begin{lemma}\label{lemma:w_set}
The value of the loss scale $w$ which equates the Fisher information number of the asymptotic distributions, $p_{LL}$ and $p_{GB,w}$, is
\begin{equation}\label{eq:w_set}
w =  \frac{ \mathrm{tr} \left\{ J(\theta_0) I(\theta_0)^{-1}  J(\theta_0)^\T \right\} }{ \mathrm{tr} \left\{ J(\theta_0) \right\} }.
\end{equation}
\end{lemma}

\begin{proof}
For $z \sim N_d(0, \Sigma)$, if $p(z)$ is the density of $z$, then the Fisher information number is $K(p) = \mathrm{tr} 
( \Sigma^{-1} )$.  Applying this result for the general Bayesian asymptotic posterior, (\ref{gbw}),
we have $K(p_{GB,w}) = \mathrm{tr} \{ w J(\theta_0) \}$. Similarly, for the loss-likelihood bootstrap asymptotic posterior we have $K(p_{LL}) = \mathrm{tr} \{  J(\theta_0) I(\theta_0)^{-1} J(\theta_0) \}$. Equating these expressions gives the required result.
\end{proof}

In practice $\theta_0$ is unknown so the empirical risk minimizer $\hat{\theta}_n$ can be used as a strongly consistent estimator of $\theta_0$. Also, as $F_0$ is unknown, matrices $I$ and $J$ can be estimated empirically.
\begin{equation}\label{eq:w_plug}
\widehat{w} =  \frac{ \mathrm{tr} \left\{ J_n(\hat{\theta}_n) I_n(\hat{\theta}_n)^{-1} J_n(\hat{\theta}_n)^\T \right\} }{ \mathrm{tr} \left\{ J_n( \hat{\theta}_n ) \right\} },
\end{equation}
where
\begin{equation*}
J_n(\theta) = \frac{1}{n} \sum\limits_{i=1}^n \nabla^2 \ell(\theta, x_i) \quad\mbox{and}\quad
 I_n(\theta) = \frac{1}{n} \sum\limits_{i=1}^n \left\{ \nabla \ell(\theta, x_i) \nabla \ell(\theta, x_i)^\T \right\}.
\end{equation*}
This methodology can be applied to a parameter defined by an arbitrary loss function. For the case where a self-information loss is used, and the data is generated from the associated sampling distribution for some value of $\theta_0$, we still recover Bayes' theorem. This is shown in the following lemma; for which regularity conditions can be found in the Supplementary Material.

\begin{lemma}\label{lemma:w_0}
If $f_0(x) = \exp \{\, - w_0\, \ell( \theta_0, x) \,\}$ for some $\theta_0 \in \Theta$, $w_0 >0$, then $w = w_0$.
\end{lemma}
\begin{proof}
We recall the result that for a regular density $f_{\theta_0}$,
\begin{equation*}
- \int_{\Omega} \nabla^2 \log f_{\theta_0}(x) \,dF_0(x) \,=\, \int_{\Omega}  \nabla \log f_{\theta_0}(x) \nabla \log f_{\theta_0}(x)^\T \,dF_0(x),
\end{equation*}
where $F_0$ has density $f_0 = f_{\theta_0}$. This immediately gives us $J(\theta_0) = w_0 I(\theta_0)$. Plugging this into (\ref{eq:w_set}) gives us $w=w_0$.
\end{proof}

\section{Illustrations}

\subsection{Normal model with quadratic loss}
Suppose we observe independent data from a multivariate normal distribution, $x_i \sim N_p(\theta_0, \Sigma_0)$,\, $i=1,\ldots,n$, and we have a loss function that is of quadratic form,
\begin{equation*}
\ell( \theta, x) = \frac{1}{2} ( x - \theta )^\T \Sigma_1^{-1} ( x - \theta).
\end{equation*}
Such a quadratic loss is non-trivial, as it could be used to estimate multiple integrals of the type
$\int h_j(x)\,d F_0(x)$, where in the loss function above we would take $x\to h(x)$.

We shall assume throughout that $\Sigma_0$ and $\Sigma_1$ are strictly positive-definite matrices. Differentiating under the integral sign shows that the parameter of interest is the population mean $\theta_0$, regardless of $\Sigma_0$ and $\Sigma_1$. The loss covariance $\Sigma_1$ is free to be set by the practitioner and does not impact the parameter of interest or the loss-likelihood bootstrap, though does change the general Bayesian posterior.

The loss-likelihood bootstrap amounts to repeatedly drawing Dirichlet weights $w_{1:n} \sim \mathrm{Dirichlet}(1,\ldots,1)$ and computing the weighted mean $\tilde{\theta}_n = \sum_i w_i x_i$. Using standard properties of the Dirichlet distribution it can be shown that this posterior is centred at the sample mean with a covariance matrix given by
\begin{equation*}
\mathrm{var} \left( \tilde{\theta}_n \mid x_{1:n} \right) = x_{1:n}^\T \left\{ \frac{1}{n(n+1)} I_{n\times n} - \frac{1}{n^2 (n+1)} 1_{n \times n} \right\} x_{1:n},
\end{equation*}
where $I_{n\times n}$ is the $n$-dimensional identity matrix and $1_{n \times n}$ is the $n \times n$ matrix with every element equal to $1$.  Given we know the distribution of $x$, we can compute the expectation of these quantities,
\begin{equation*}
E \left\{ E \left( \tilde{\theta}_n \mid x_{1:n} \right) \right\}= \theta_0, \qquad 
E \left\{ \mathrm{var} \left( \tilde{\theta}_n \mid x_{1:n} \right) \right\} = \frac{n-1}{n(n+1)} \Sigma_0.
\end{equation*}
Now let us consider the general Bayesian approach. If $\Sigma_0 = \Sigma_1$ then the loss function is equal, up to a constant, to the negative log likelihood associated with the data-generating mechanism. We say that this loss is well-specified, and find from Lemma \ref{lemma:w_set}, that we should set $w=1$. For this special case, general Bayesian updating coincides with Bayesian updating under a normal location model.

Given our knowledge of the distribution of the data, we can determine the matrices $I$ and $J$ using standard properties of the normal distribution. In particular, we have $I(\theta_0) = \Sigma_1^{-1} \Sigma_0 \Sigma_1^{-1}$ and $J(\theta_0) = \Sigma_1^{-1}$. Using Lemma \ref{lemma:w_set} we get the following expression for the loss scale that calibrates the asymptotic posterior Fisher information number,
\begin{equation}\label{eq:w_normal}
w =  \frac{ \mathrm{tr} ( \Sigma_0^{-1} ) }{ \mathrm{tr} ( \Sigma_1^{-1} ) }.
\end{equation}
Suppose that the loss function is the standard quadratic loss ($\Sigma_1 = I_p$), and each dimension of $x$ is independent, with $\Sigma_0 = \mathrm{diag}( \sigma_1^2, \ldots, \sigma_p^2)$, the $p \times p$ diagonal matrix with diagonal elements $(\sigma_1^2,\ldots,\sigma_p^2)$. Plugging these expressions into (\ref{eq:w_normal}) we obtain $w= p^{-1} \sum_{i=1}^p \sigma_i^{-2}$. That is, the loss scale is the average precision. If $\max_i \sigma_i^2 < 1$, then the data is under-dispersed relative to the loss (considering the loss as a negative log likelihood). The loss scale calibration acts to correct the loss likelihood to better match the likelihood of the well-specified model. Specifically, it is clear we will obtain a $w>1$, thus reducing the variance associated with our loss likelihood. Similarly if $\min_i \sigma_i^2 > 1$ then we obtain $w<1$ which calibrates the loss to better match the data. We do not expect in all cases for our loss likelihood to match the likelihood under the well-specified model, as the data's covariance structure of our data has many more degrees of freedom than the scalar $w$ available to calibrate the general Bayesian posterior. Reassuringly, Lemma \ref{lemma:w_0} tells us that if $\Sigma_1 = w_0 \Sigma_0$ we do obtain correct calibration; i.e. $w=w_0$.

More generally, given matrices $\Sigma_0$ and $\Sigma_1$ we obtain the following general Bayesian posterior:
\begin{equation*}
p_{GB} ( \theta \mid x_{1:n} ) \ \propto \ p(\theta)\, \exp \left[\, - \frac{1}{2} (\theta - \bar{x} )^\T  \left\{ \frac{ n \, \mathrm{tr}( \Sigma_0^{-1} )}{  \mathrm{tr}( \Sigma_1^{-1} )} \Sigma_1^{-1} \right\} ( \theta - \bar{x} ) \, \right],
\end{equation*}
where $\bar{x} = n^{-1} \sum_{i=1}^n x_i$. The scale of the general Bayes loss is set to match the trace of the precision matrix of the loss likelihood with that of the data-generating mechanism. For univariate data this setting of $w$ ensures that the loss likelihood is equal to the likelihood, and thus the general Bayesian update will coincide with the Bayesian update.

In practice the distribution of the data will be unknown. In this case we still have $J(\theta) = \Sigma_1^{-1}$ for all $\theta$, and $I(\theta) =  \int \Sigma_1^{-1} (x - \theta) (x - \theta )^\T \Sigma_1^{-1} dF_0(x)$. As $\theta_0$ is the population mean, we can see that $I( \theta_0 ) = \Sigma_1^{-1} \, \mathrm{C}   \,\Sigma_1^{-1}$, where $\mathrm{C} $ is the population covariance matrix for $x$. So we can rewrite $w$ as
$w =  \mathrm{tr} (  \mathrm{C}^{-1} ) /\mathrm{tr} ( \Sigma_1^{-1} )$.
If $p=1$ , an unbiased estimator $\widehat{w^{-1} }$ of $w^{-1}$ is the sample variance divided by the loss variance. In this case, $\widehat{w} = 1 / \widehat{w^{-1} }$ may be a biased but consistent estimator of $w$. The plug-in estimator of (\ref{eq:w_plug}) is available for $p>1$.

\subsection{Bayesian support vector machine for binary classification}

Consider the problem of binary classification where we observe $ \{ x_i = (y_i, z_i) \in \{-1,1\} \otimes \mathbb{R}^p \ ; \ i=1,\ldots,n \}$, with each $x_i \sim F_0$, where $z_i$ denotes the covariates of an observation belonging to class $y_i$. We would like to predict future $y^*$s given their covariates $z^*$. Specifically, the objective is to learn about the optimal linear classification rule that minimizes the expected loss under $\phi$, a margin-based loss function,
\begin{equation*}
(\alpha_0, \beta_0) = \arg \min_{\alpha,\beta} \int \phi \left\{ y \left( \alpha + \beta^\T z \right) \right\} \,dF_0(y,z).
\end{equation*}
In a general Bayesian framework, we can compute a posterior distribution for the parameter of interest $\theta = (\alpha,\beta)$ given just the loss function, prior beliefs for $\theta$ and a loss scale. Our work in this paper provides a means for determining the loss scale. The loss-likelihood bootstrap also provides a posterior sample without requiring a loss scale, however it does not admit a prior. 

Often the loss function of interest is the $0-1$ loss, $\phi_{0-1}(\xi) = \mathbb{1}( \xi \leq 0 )$, whose non-convexity leads to a number of computational problems relating to optimization under this loss. A popular approach is to use a convex surrogate loss in place of $\phi_{0-1}$. \citet{Bartlett2006} explore this idea formally, and provides some simple conditions for a surrogate loss to be Bayes-risk consistent to the $0-1$ loss.

A popular classification method in the machine learning literature is the support vector machine \citep{Cortes1995a}. Applied to linear classifiers as considered above, this method amounts to penalized optimization of a convex surrogate loss:
\begin{equation*}
(\widehat{\alpha}, \widehat{\beta}) = \arg \min_{\alpha,\beta} \sum\limits_{i=1}^n \phi_0 \left\{ y_i \left( \alpha + \beta^\T z_i \right) \right\} + \ \frac{\lambda}{2}   \parallel \beta \parallel^2.
\end{equation*}
Here $\lambda \geq 0$ is a hyperparameter and $\phi_0( \xi ) = \max( 0 , 1 - \xi )$ is often referred to as the hinge loss. The optimization can be solved efficiently as a convex quadratic programming problem. The output is a linear classification rule that performs well empirically on a wide range of classification problems. The method doesn't provide any uncertainty quantification about the optimal linear boundary. We consider how we can use the ideas developed in this paper to provide uncertainty quantification about the optimal linear discriminant.

The problem with using the hinge loss in our framework is that we use first and second derivatives of the loss in the calculation of $w$. The hinge loss can easily be smoothed however; see for example \citet{zhang2004solving}. We construct a smoothed hinge loss $\phi_2$ that coincides with the first three derivatives of $\phi_0$ outside of $(0,1)$ and is a monotonic polynomial in between:

\begin{equation*}
\phi_2( \xi ) =  \begin{cases}
\frac{1}{2} - \xi, & \mathrm{if} \ \xi < 0 \\
\xi^6 - 3 \xi^5 + \frac{5}{2} \xi^4 - \xi + \frac{1}{2}, & \mathrm{if} \ \xi \in [0,1] \\
0 & \mathrm{otherwise.}
\end{cases} 
\end{equation*}
A comparison of this loss function to the standard hinge loss can be found in Fig. \ref{fig:plot_synthetic_1}.

We can sample from the loss-likelihood bootstrap or the calibrated general Bayes posterior of the parameter that minimses the expectation of $\ell(\alpha, \beta, x ) = \phi_2 \left\{ y \left( \alpha + \beta^\T z \right) \right\}$. We propose the following prediction routine for a new observation $\widehat{y}$ given the covariates of the new observation $z^*$, and a posterior sample. For each element of the posterior sample $(\alpha_i, \beta_i)$ we compute the class prediction $\widehat{y}_i = \mathrm{sign}( \alpha_i + \beta_i^\T z^* )$. We then predict the modal class across all of these predictions.

To test out this classification method, we constructed a synthetic dataset from the following distribution:
\begin{equation*}\label{eq:synthetic}
\mathbb{P}( y = 1) = \frac{1}{2} = \mathbb{P}( y = -1 ), \quad \mbox{with}\quad p(z \mid  y) = N(z\mid  y , 1).
\end{equation*}
A class conditional density plot can be found in Fig. \ref{fig:plot_synthetic_1}. A linear classification boundary attains the Bayes risk under the $0-1$ loss. It is of the form $\widehat{y} = 2 \,\mathbb{1}(z > 0) - 1$, which implies that $\alpha=0$ and $\beta>0$. Our classification posterior should concentrate here for large sample sizes.

\begin{figure}[!htbp]
\begin{center}
   \includegraphics[scale=1]{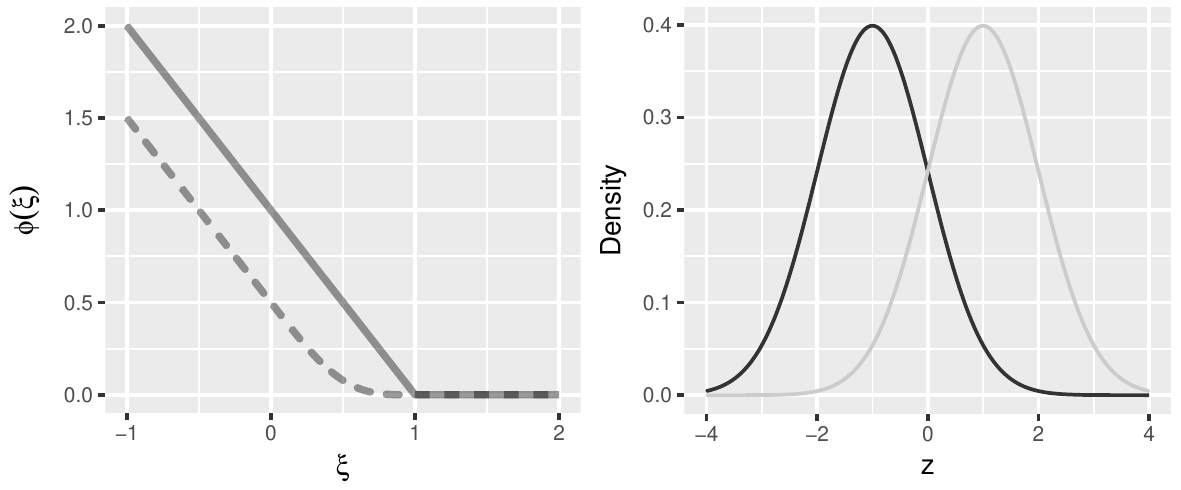}
   \captionsetup{justification=centering}
   \caption{Left: Margin loss plot for hinge (solid) and smooth hinge loss (dashed). Right: Class conditional probability density plot for the synthetic example.}
   \label{fig:plot_synthetic_1}
\end{center}
\end{figure}

We generated a synthetic dataset of size $n=100$ and assigned independent $N(0, 10^2)$ priors to $\alpha$ and $\beta$. For the general Bayes posterior we computed a loss scale using the plug-in estimator $\widehat{w}$ of $w$ specified in (\ref{eq:w_plug}), under the smooth hinge loss $\phi_2$. We generated a posterior sample using the Hamiltonian Monte Carlo routine implemented in probabilistic programming language Stan \citep{Carpenter2016}. We also generated a sample from the loss-likelihood bootstrap. Figure \ref{fig:plot_synthetic_2} shows the joint general Bayesian posterior density of $(\alpha, \beta)$ estimated using a posterior sample of size $10\,000$, as well as marginal density plots for both methods. The general Bayesian posterior distribution matches well the loss-likelihood bootstrap, though the general posterior is calibrated to match asymptotic posterior information, as opposed to coverage.

\begin{figure}[!htbp]
\begin{center}
   \includegraphics[scale=0.85]{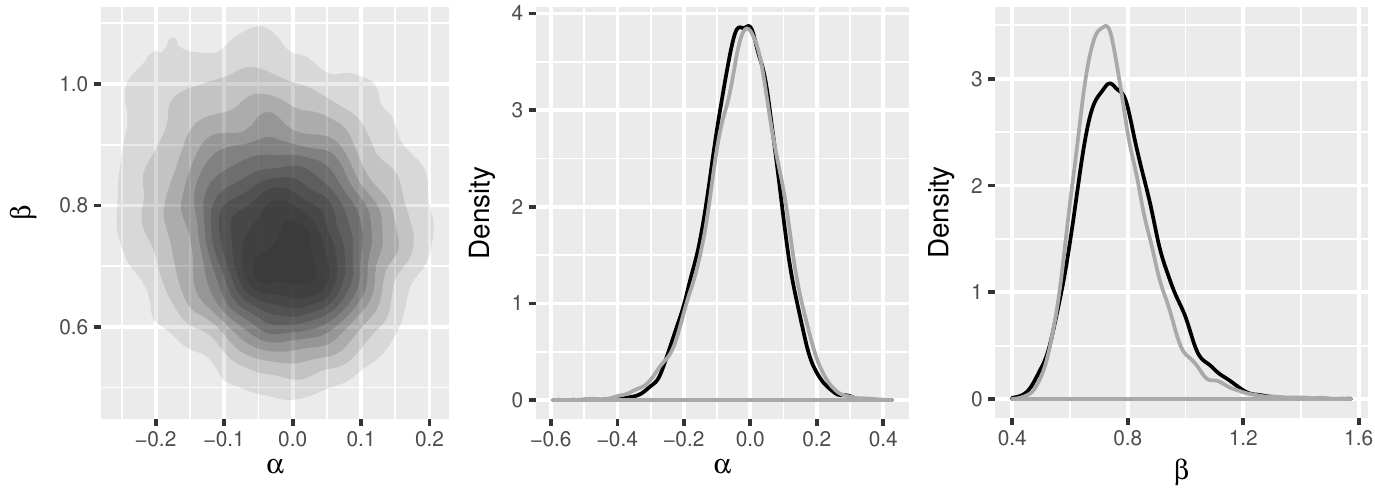}
   \captionsetup{justification=centering}
   \caption{Synthetic example plots. Left: Joint general Bayesian posterior density plot for $(\alpha,\beta)$. Middle: Marginal density plot for $\alpha$ for the general Bayes posterior (black) and the loss-likelihood bootstrap (grey). Right: Marginal density plot of $\beta$ for the general Bayes posterior (black) and loss-likelihood bootstrap (grey).}
   \label{fig:plot_synthetic_2}
\end{center}
\end{figure}

Figure \ref{fig:plot_synthetic_3} shows the general posterior probability that $\alpha + \beta z > 0$ given $z$, as a function of $z$, for various sample sizes. The curves are the average over $100$ repetitions. The loss scale estimate $\widehat{w}$ and the misclassification error on a test dataset of $10\,000$ samples, relative to the misclassification error of a linear support vector machine, are displayed in Fig. \ref{fig:plot_synthetic_3}. We used the e1071 support vector machine implementation in R, with five-fold cross validation to set the regularization parameter. 

Posterior samples of size $1000$ were used for prediction, and for the general Bayesian posterior the first $1000$ samples were discarded as burn-in. In the left pane of Fig. \ref{fig:plot_synthetic_3} we see the posterior predictive probability curves steepening as the number of observations increase, at $z=0$, showing that the posterior mass does concentrate at the optimal classification boundary. The plug-in loss scale $\widehat{w}$ exhibits higher variance for small datasets. Additionally our plug in estimator $\widehat{w}$ exhibits some bias which diminishes with the number of observations. The performance of our classification routines is very similar to a linear support vector machine.

\begin{figure}[!htbp]
\begin{center}
   \includegraphics[scale=0.85]{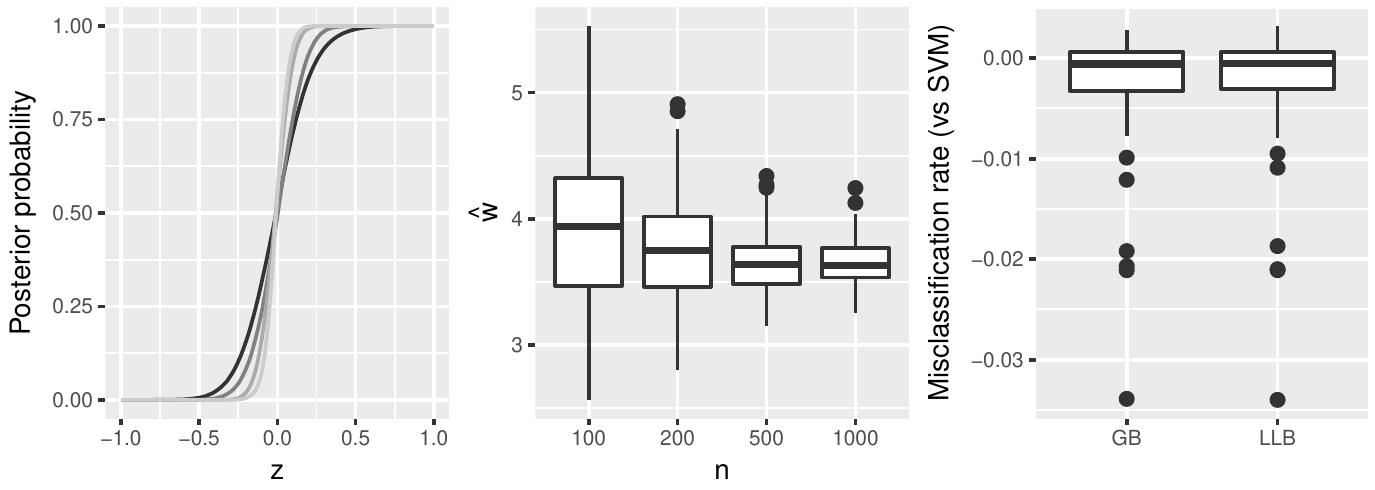}
   \captionsetup{justification=centering}
   \caption{Further synthetic example plots. Left: General Bayesian predictive probability of optimal decision rule predicting class 1, as a function of the covariate $z$, for $n=100,200,500,1000$ (from dark to light). Middle: Box plot for $\widehat{w}$ as a function of sample size, over 100 repeated runs. Right: Box plot for each method, of its misclassification rate minus the misclassification rate of a linear support vector machine, for $n=100$ observations, over $100$ repetitions; GB and LLB refer to general Bayesian and loss-likelihood bootstrap respectively.}
   \label{fig:plot_synthetic_3}
\end{center}
\end{figure}

As a more challenging example, we took the Statlog German Credit dataset from the UCI Machine Learning Depository (http://archive.ics.uci.edu/ml), preprocessed following \citet{Fernandez-Delgado2014}. It contains 24 covariates across 1000 customers, where each customer belongs to one of two classes pertaining to the lending experience of the customer.

The following test was repeated $100$ times: we randomly split our dataset 75:25 to obtain a training and test dataset. For the general Bayesian method, we use an independent $N(0,100)$ prior  on each covariate dimension and the intercept. We computed $\widehat{w}$ and generated a sample of size $B=1000$ in the same way as in the synthetic dataset example. Misclassification rate was recorded in comparison to that of a linear support vector machine.

The classification performance of our general Bayesian procedure is very similar to that of the support vector machine method that does not provide uncertainty quantification. Posterior marginals for $\alpha$ and the first component $\beta_1$ of $\beta$ are displayed in Fig. \ref{fig:german_box} for a single run. They seem to be well aligned again.

\begin{figure}[!htbp]
\begin{center}
   \includegraphics[scale=0.85]{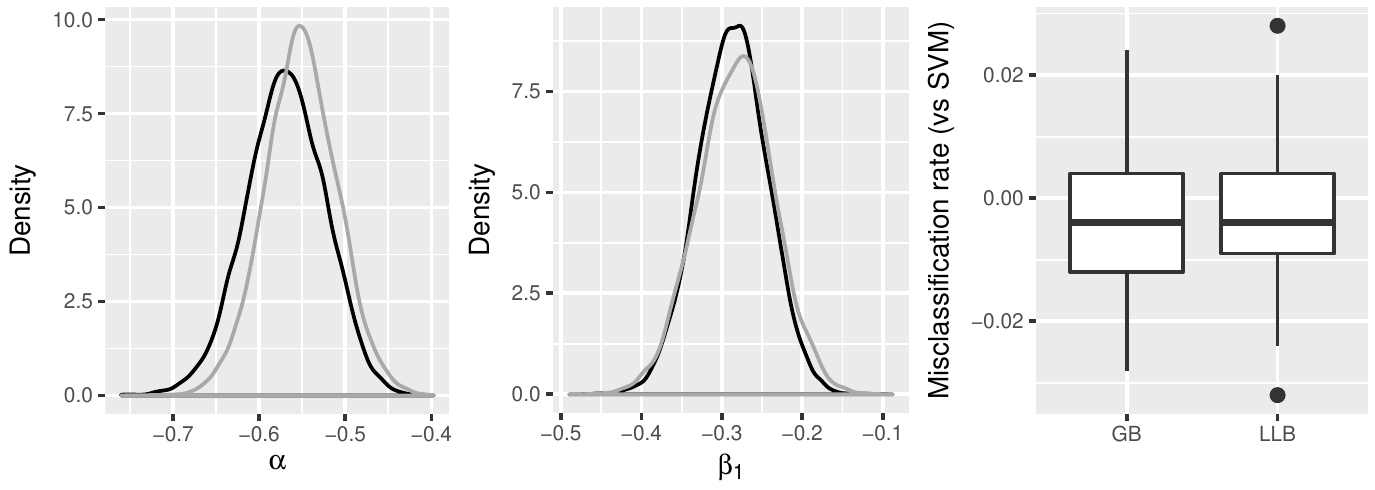}
   \captionsetup{justification=centering}
   \caption{Statlog German Credit plots. Left: Marginal density plot for $\alpha$ for the general Bayes posterior (black) and the loss-likelihood bootstrap (grey). Middle: Marginal density plot for $\beta_1$ for the general Bayes posterior (black) and the loss-likelihood bootstrap (grey). Right: Box plot of misclassification rate minus the support vector machine misclassification rate; GB and LLB refer to general Bayesian and loss-likelihood bootstrap respectively.}
   \label{fig:german_box}
\end{center}
\end{figure}

Both of our methods can easily be adapted to cost-sensitive classification problems, by use of an asymmetric hinge loss \citep{Scott2012}.

\section{Discussion} In this paper we have presented two methods for generating samples from posterior belief distributions about a parameter defined by a loss function. Both methods make minimal assumptions about the data-generating mechanism.

The loss-likelihood bootstrap is a prior-free method built on a Bayesian nonparametric model. It does not require any calibration. Computationally, the method centres around optimization. Independent samples are generated and the method is trivially parallelizable.

The loss-likelihood bootstrap does not admit a prior specification, which may be unappealing to the subjective Bayesian. 
One idea would be to importance-weight the Bayesian bootstrap samples with respect to a chosen prior. This is also discussed in \cite{Newton1994} as a means of aligning weighted likelihood bootstrap samples to a Bayesian posterior, though it requires density estimation under the Bayesian bootstrap, which may well suffer from a curse of dimensionality in even moderate dimensions. Furthermore, these weights are prone to degeneracy if the posteriors under the two priors are meaningfully different. \citet{Kessler2015} suggests constructing a prior on the state probabilities that matches the marginal prior for the parameter of interest, and is conditionally non-informative given the parameter.  Again, to compute this conditionally non-informative prior requires some density estimation.

General Bayesian updating provides a decision-theoretic  posterior under a coherency condition and a proper prior. It does require the calibration of data information relative to the prior information, for which we have provided a new method using the loss likelihood bootstrap. 

Our calibration argument relies on the matching of asymptotic posterior information, which is well motivated as both methods make few distributional assumptions and target the same parameter. In \cite{Bissiri2016} a number of possible alternative means of calibration are discussed, including subjective calibration, unit information loss matching and hierarchical treatment via a loss function on $w$. Although these ideas may suit particular applications, they are less well motivated in general, compared to the proposal in this paper. \cite{Syring2017} provide an iterative method for setting $w$ by ensuring calibration of a single user-specified posterior credible region. However, this method is computationally demanding and switches the issue of setting the loss scale to one of choosing a single credible region to calibrate.

For the special case of a potentially misspecified self-information loss, \citet{Holmes2017} set the loss scale by matching the expected information gain from a single observation, for two experiments. In the first experiment the self-information loss is well specified in relation to the distribution of the observations. In this case the Bayesian update is optimal, so $w=1$. In the second experiment the data come from an unknown distribution so a general Bayesian update with self-information loss is used. For self-information losses $w$ acts as a tempering parameter on the likelihood. The Fisher information distance is used to measure the prior-to-posterior information gain. The data is used to estimate these quantities.

The approach of \cite{Holmes2017} focuses on prior information, whereas our method calibrates asymptotic posterior information to a bootstrap. One important quality of our method is that it recovers the parametric Bayesian learning rate when the model is correct up to an arbitrary tempering (any $w_0 > 0$ in Lemma \ref{lemma:w_0}), whereas \cite{Holmes2017} will only recover the correct rate if the correct log likelihood is used; i.e $w_0=1$.

However, \cite{Holmes2017} is not an option in the general loss case as we do not have a benchmark experiment for which $w$ is known. The criterion of Fisher information matching with the loss-likelihood bootstrap is applicable to calibrating posterior distributions based on loss functions, under minimal assumptions about the underlying data generating mechanism. This includes self information loss, $\ell(\theta,x)=-\log p(x\mid \theta)$, as a particular case.

\bibliographystyle{plainnat}
\bibliography{w-setting}

\end{document}